\documentclass[twocolumn,superscriptaddress,aps,longbibliography]{revtex4-2}
\usepackage{graphicx}
\usepackage{tikz}
\usepackage{amsmath,amsthm,amssymb,latexsym,array,amsfonts}
\usepackage{hyperref}
\usepackage{dsfont,mathrsfs,verbatim,bbold}
\usepackage{bm,times}
\usetikzlibrary{positioning,arrows.meta}

\usepackage{adjustbox}
\usepackage{booktabs}

\newcommand{\ket}[1]{\left|#1\right>}
\newcommand{\abs}[1]{\left|#1\right|}

\newtheorem{theorem}{Theorem}
\newtheorem{proposition}[theorem]{Proposition}

\newtheorem{definition}[theorem]{Definition}

\begin{document}

\title{Neural Information Causality: A Random-Access Capacity Law for Representation Learning}

\author{Jeongho~Bang}\email{jbang@yonsei.ac.kr}
\affiliation{Institute for Convergence Research and Education in Advanced Technology, Yonsei University, Seoul 03722, Republic of Korea}
\affiliation{Department of Quantum Information, Yonsei University, Incheon 21983, Republic of Korea}

\author{Marcin~Paw\l{}owski}\email{marcin.pawlowski@ug.edu.pl}
\affiliation{International Centre for Theory of Quantum Technologies, University of Gda\'{n}sk, 80-308 Gda\'{n}sk, Poland}
\affiliation{Institute of Theoretical Physics and Astrophysics, Faculty of Mathematics, Physics and Informatics, University of Gda\'{n}sk, 80-308 Gda\'{n}sk, Poland}

\date{\today}

\begin{abstract}
How can a learned representation remain useful when the question it must answer is chosen only after encoding? Preserving information is essential, but does not by itself ensure access to the requested target. We address this problem by deriving ``neural information causality (Neural-IC),'' a representation law that connects causal access to the physically specified interface capacity. For a fresh database of independent bits, the aggregate information recoverable across later queries cannot exceed that capacity, even with pre-established quantum resources. A decomposition separates capacity slack, information organized jointly across targets, and losses in producing answers. Exhaustive evaluation of trained binary codes reveals why nearly full use of a memory's information budget can coexist with limited retrieval. In a nested Bell-correlation construction, quantum assistance improves average retrieval over the optimal classical one-bit benchmark within the same communicated budget. These results give representation learning a design principle that connects information capacity to the organization and extraction of target information. By locating the constraints on later answers, the framework can inform memory allocation, encoding objectives, and decoder outputs, providing a basis for seeking better retrieval within a fixed budget.
\end{abstract}

\maketitle

\section{Introduction}\label{sec:introduction}

Consider a model that stores a database in an intermediate representation and later answers a query about one of its entries. If the requested entry is known during encoding, the model can retain that target alone. When the query arrives only after encoding, the same representation must support several possible targets. This requirement matters in representation learning because a stored code is often reused rather than formed anew for each requested answer~\cite{Bengio2013RepresentationLearning}. Compressed memories and encoder--decoder systems preserve information before its later use is specified, while trainable external memories separate writing from subsequent reading~\cite{Graves2014NTM,Graves2016DNC}. Their usefulness therefore depends on which future queries the retained information can support after access to the original data ends. The question is how much a decoder can recover across those queries through a fixed interface, and how the learned code organizes that information. We study this question by specifying the information paths.

Our formulation, termed ``neural information causality (Neural-IC),'' treats the representation as the data-dependent interface between encoding and a later query. The parameters are fixed before a fresh database is sampled, and any pre-established resource is initially independent of that database. The decoder receives current data only through the specified interface. Under this restriction, random-access coding and information causality supply an information bound on the answers across the query set~\cite{Pawlowski2009}. Applying that bound requires a physical capacity model for the representation. A coordinate with binary values, a quantized real coordinate, and an unrestricted ideal real coordinate need not carry the same amount of information. Finite alphabet, precision, or noise constraints instead determine a budget before retrieval is evaluated. Applying the same physical constraint to designed and learned maps gives a common measure of their available information. The remaining question is how much of that information the answers expose. A chain-rule decomposition locates the difference in information retained jointly about several targets and information lost during answer extraction. The information bottleneck approach selects or regularizes a representation through information about an input and a designated target~\cite{Tishby2000IB,Alemi2017VIB}. Here, the target is selected after encoding, so the organization of information across possible targets becomes part of the description. Ordered conditioning extends this accounting to correlated targets, while an additional memory contribution specifies how episode-dependent parameters or another data path change the available budget.

The analysis thus separates limits on the amount of retained information from limits on how much individual target answers can expose. Exhaustive calculations for small trained classical binary encoder and decoder models resolve this distinction even when their retained information approaches the interface budget. Quantum assistance provides a complementary example of how access can change at fixed communicated capacity. In the nested Clauser--Horne--Shimony--Holt (CHSH) construction~\cite{Clauser1969CHSH,Pawlowski2009}, it improves average retrieval accuracy over the optimal classical one-bit benchmark. At the Tsirelson limit~\cite{Tsirelson1980}, the aggregate information remains finite although each answer becomes less informative as the database grows. The quantum example and learned codes therefore address different parts of the same question, namely how a representation makes its retained information available across queries.

For representation design, these distinctions connect the physical interface to the learning task. The encoder organizes information relative to the target family, while the decoder and its output alphabet determine what each answer exposes. Their separate contributions can guide choices about precision, memory allocation, and training objectives under a common budget. Within query separation, Neural-IC thus provides a basis for seeking better access to retained information and for assessing whether a proposed change addresses the relevant limitation.

\section{Causal information structure of a learned representation}\label{sec:neural-ic-framework}

We consider representations whose role is to preserve data for a later, as yet unspecified query. The relevant restriction is causal rather than geographical. The encoder and decoder may belong to one computing system, but their access to the current data must be separated by an interface. The learning selects the encoding and decoding maps. The information bound constrains the resulting maps once they have been fixed. All information quantities are evaluated for a fixed model realization $W$. If shared parameters are randomized between episodes, we condition on that draw or include it in $B$. The independence of random weights alone would not permit omitting them from the decoder information. All logarithms are base two, and $h(p)=-p\log p-(1-p)\log(1-p)$ denotes binary entropy.

\begin{definition}[Query-separated representation]\label{def:query_separated}
A query-separated architecture has maps
\begin{eqnarray}
H \leftarrow \mathrm{Enc}(\vec a;R_A,W),
\quad
\beta \leftarrow \mathrm{Dec}(H,B,b;W),
\label{eq:query_separated_maps}
\end{eqnarray}
where the representation $H$ is formed before the query $b$ is supplied. The decoder has no other access to the database $\vec a$. The parameters $W$ and the pre-established resource $(R_A,B)$ are independent of the fresh database and query. Any data-dependent state retained between encoding and decoding belongs to the interface.
\end{definition}

\begin{figure}[t]
\centering
\includegraphics[width=1.00\columnwidth]{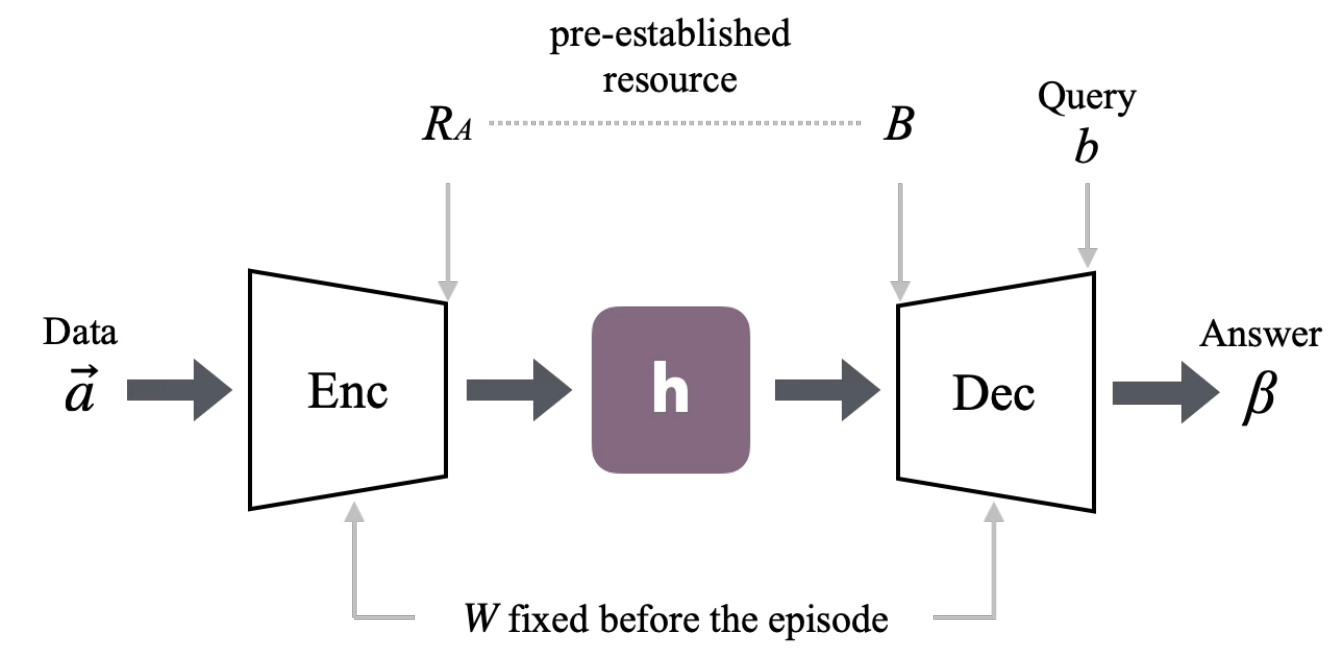}
\caption{Causal structure of a query-separated representation. $R_A$ and $B$ denote the encoder-side and decoder-side systems of a pre-established resource. They may share classical randomness
or quantum entanglement, but their joint initial state and the fixed parameters $W$ are independent of the fresh database $\vec a$ and query $b$. The encoder forms $H$ before the query is supplied to the decoder, which already holds $B$. The dashed link denotes pre-established correlations. The database influences the answer through the interface $H$. Any additional episode-dependent state available to the decoder must be included in the information model.}
\label{fig:causal-structure}
\end{figure}

The same definition covers trained and designed maps. In a learned model, the weights may reflect its training distribution, but they are fixed before the database for an evaluation episode is sampled. Otherwise, they would constitute an additional data-dependent representation. Fig.~\ref{fig:causal-structure} shows the counted information paths. Likewise, an external memory consulted after the query is known must be included in the counted interface. The hidden-layer dimension alone does not specify the information available to the decoder.

\begin{definition}[Neural random access coding]\label{def:neural_RAC}
In Neural-RAC, $\vec a=(a_0,\ldots,a_{N-1})$ consists of independent unbiased bits, and $b$ is an independent, uniformly sampled index. A query-separated architecture outputs $\beta\in\{0,1\}$ as an estimate of $a_b$. Its per-query accuracy and information score are
\begin{eqnarray}
P_K &:=& \Pr[\beta=a_K\mid b=K],\label{eq:neural_Pk_def} \nonumber \\
I_{\mathrm{N\text{-}RAC}} &:=& \sum_{K=0}^{N-1}I(a_K:\beta\mid b=K).
\label{eq:INRAC_def}
\end{eqnarray}
\end{definition}

The score uses only observed classical variables. It compares the different query branches over repeated episodes; it does not presume that all answers can be retrieved simultaneously from one representation. Its definition and its connection to accuracy are inherited from the random-access formulation~\cite{Pawlowski2009}. For each branch, the error bit $e_K=a_K\oplus\beta$ gives $H(a_K\mid\beta,b=K)\le H(e_K\mid b=K)=h(P_K)$. Consequently, we have
\begin{eqnarray}
I_{\mathrm{N\text{-}RAC}} \ge N-\sum_{K=0}^{N-1}h(P_K).
\label{eq:INRAC_lower_bound}
\end{eqnarray}
The qquality holds when every induced target-to-answer channel is binary symmetric. Without that symmetry, accuracy supplies a lower bound, and the score itself is obtained from the joint target--answer distributions.

\subsection{Representation law and independently bounded capacity} 

For nonclassical side resources, the information quantities involving $(H,B)$ require an information calculus. We use the assumptions underlying IC: symmetric nonnegative mutual information, agreement with Shannon information on classical registers, data processing under local maps, and a chain rule with nonnegative conditional information. Independent systems have zero mutual information; for a classical register $X$, $I(X:XY)=H(X)$. Classical and quantum information obey these properties. They need not hold for arbitrary no-signaling theories~\cite{Pawlowski2009}.

\begin{theorem}[Representation information law]\label{thm:neural-ic-embedding}\label{thm:Neural_IC}
Under these assumptions, a query-separated Neural-RAC architecture obeys
\begin{eqnarray}
I_{\mathrm{N\text{-}RAC}} \le I(\vec a:H,B).
\label{eq:embedding-bound-revised}
\end{eqnarray}
If its interface model independently bounds
$I(\vec a:H,B) \le C_H$, then
\begin{eqnarray}
I_{\mathrm{N\text{-}RAC}}\le C_H.
\label{eq:capacity-bound-revised}
\end{eqnarray}
In particular, a classical $m$-bit interface satisfies $I_{\mathrm{N\text{-}RAC}}\le m$, even with quantum side resources.
\end{theorem}

\begin{proof}---Let $\vec a_{<K}=(a_0,\ldots,a_{K-1})$. The chain rule and independence of the database bits give
\begin{eqnarray}
I(\vec a:H,B) &=&\sum_K I(a_K:H,B\mid\vec a_{<K}) \nonumber \\
	&=&\sum_K I(a_K:H,B,\vec a_{<K}) \nonumber \\
	&\ge&\sum_K I(a_K:H,B) \nonumber \\
	&\ge&\sum_K I(a_K:\beta\mid b=K).
\label{eq:embedding-chain-revised}
\end{eqnarray}
The first inequality discards the preceding bits. The second applies data processing to the decoder for each fixed query; query separation ensures the same pre-query system $(H,B)$ in all branches. This is the IC argument applied to the causal interface~\cite{Pawlowski2009}. The capacity consequence follows from the independently specified bound. For classical $H\in\{0,1\}^m$, independence of $\vec a$ and $B$ further gives
\begin{eqnarray}
I(\vec a:H,B) &=& I(\vec a:H\mid B) \nonumber \\
	&=& I(H:\vec a,B)-I(H:B) \nonumber \\
	&\le& I(H:H,\vec a,B)=H(H) \le m.
\label{eq:hard-bottleneck-proof-revised}
\end{eqnarray}
\end{proof}

The theorem concerns the causal information structure of the learned representation, not the procedure that learns it. A learning architecture must specify which variables can depend on the current episode, which reach the decoder, and what physical constraints bound their information. Combining these commitments with the information law yields a causal structure for representation-mediated predictions.

A convenient interface certificate is
\begin{eqnarray}
C_H(\mathsf H) := \sup_{p(\vec a),\,\mathrm{Enc}\in\mathsf H} I(\vec a:H\mid B),
\label{eq:physical-capacity-def}
\end{eqnarray}
where $\mathsf H$ specifies the admissible interface and encoders, and the database is independent of the resource. A known upper bound on this supremum can be used without evaluating it exactly. It is supplied before the retrieval experiment rather than defined by its observed score. The channel-capacity formulation of IC already permits this interpretation of a communication budget~\cite{Miklin2021}.

For a classical alphabet $\mathcal H$, $C_H\le\log|\mathcal H|$. A vector with $d$ coordinates and $q$ bits of precision per coordinate has $C_H \le dq$. For a classical decoder-side register $B$ and an additive Gaussian interface $H=U+Z$, assume $Z\sim\mathcal N(0,\sigma^2\mathds{1}_d)$ is independent of the input, database, and side resource, and $\mathbb E\|U\|^2\le dP$. Conditional data processing and the Gaussian maximum-entropy bound give
\begin{eqnarray}
C_H \le \frac{d}{2}\log(1 + P \sigma^{-2}).
\label{eq:awgn-capacity-certificate}
\end{eqnarray}
This is the usual capacity per $d$-coordinate channel use. The noise must not be accessible separately to the decoder~\cite{CoverThomas2006}. In the binary task all certificates can also be capped by
$H(\vec a)=N$. An unconstrained real coordinate can encode all $N$ bits, so its dimension supplies no capacity bound independent of database size. For example, the distinct values $H=\sum_{K=0}^{N-1}2^{-(K+1)}a_K$ permit exact recovery if precision and readout are unrestricted. The code still requires $N$ bits of information, despite occupying one coordinate. Quantizing that coordinate to $q$ bits restores the alphabet bound $C_H\le q$. Appendix~\ref{app:sponge} illustrates the distinction between output length and retained information through cryptographic sponge constructions.

Combining the information law with Eq.~(\ref{eq:INRAC_lower_bound}) yields the representation--accuracy tradeoff
\begin{eqnarray}
N - \sum_K h(P_K) \le C_H.
\label{eq:capacity_accuracy_general}
\end{eqnarray}
For query-symmetric accuracy $P_K=P$, it reduces to $N[1-h(P)]\le C_H$. The conclusion is conditional on the causal model: if the encoder sees $b$, sending $a_b$ through one bit answers every query exactly and invalidates the common pre-query interface used in the proof. Similarly, precision packing or episode-specific weights changes the information supplied to the decoder. These cases delimit the theory's applicability rather than contradict its bound.

The effect of a pre-query side memory can also be expressed quantitatively. If the decoder additionally receives an episode-dependent register $M$, the same argument and the chain rule imply
\begin{eqnarray}
I_{\mathrm{N\text{-}RAC}} &\le& I(\vec a:H,B,M) \nonumber \\
	&\le& C_H + I(\vec a:M\mid H,B),
\label{eq:side-memory-slack}
\end{eqnarray}
provided the designated interface still satisfies $I(\vec a:H,B)\le C_H$. An excess over that budget therefore requires at least the corresponding amount of additional conditional information in $M$. This consequence makes explicit how a relaxed architecture changes the representation law. In particular, storing the database in episode-specific parameters makes those parameters part of $M$. If $M=\vec a$, then $I(\vec a:M\mid H,B)=N-I(\vec a:H,B)$, and the expanded carrier can contain the whole database. Parameters learned on previous episodes and fixed before the current database is drawn do not supply this additional memory.

\subsection{Information retained and information exposed by answers}\label{sec:information-gap} 

The representation law identifies two information quantities that need not coincide. The joint system $Y=(H,B)$ may retain information about several targets together even when no individual-target answer reveals it. A chosen decoder can lose further information. These effects have an exact chain-rule decomposition for the independent-bit task.

\begin{proposition}[Representation--retrieval gap]\label{prop:information-gap}
For a query-separated representation of independent database bits, define $\Delta=I(\vec a:Y)-I_{\mathrm{N\text{-}RAC}}$. Then,
\begin{eqnarray}
\Delta &=& \sum_K I(a_K:\vec a_{<K}\mid Y) \nonumber \\
	&& + \sum_K\left[ I(a_K:Y)-I(a_K:\beta\mid b=K) \right].
\label{eq:information-gap-decomposition}
\end{eqnarray}
Both sums are nonnegative under the information assumptions of Theorem~\ref{thm:Neural_IC}.
\end{proposition}

\begin{proof}---The independence gives $I(a_K:\vec a_{<K})=0$. Applying the chain rule in two orders therefore yields
\begin{eqnarray}
I(a_K:Y\mid\vec a_{<K}) = I(a_K:Y) + I(a_K:\vec a_{<K}\mid Y).
\label{eq:gap-chain-identity}
\end{eqnarray}
Summing this identity and subtracting the retrieval score gives Eq.~(\ref{eq:information-gap-decomposition}). Nonnegative conditional information bounds the first sum. Data processing for each decoder
branch bounds the second.
\end{proof}

The first term measures dependencies among initially independent targets after conditioning on the representation and side resource. It records joint information that individual-target queries do not
expose. The total of this term is independent of the chosen target ordering, although its separate summands can change. The second term is a loss from the available single-target information to the actual
answers. It depends on the decoder family as well as on the representation. For quantum $Y$, this difference compares information before decoding with classical information in the output. It does not
assume that the original quantum system survives a measurement.

A two-bit parity code makes the distinction concrete. Let $H=a_0\oplus a_1$ and omit side resources. The representation carries one bit about the database, yet it is independent of either target. Every answer obtained only from $H$ has zero information about an individual bit, so $I_{\mathrm{N\text{-}RAC}}=0$. In Eq.~(\ref{eq:information-gap-decomposition}), $I(a_1:a_0\mid H)=1$ accounts for the entire gap. Conversely, storing $a_0$ and returning it for query zero exposes the full one-bit score. The interface sizes are identical. Their information is organized differently relative to the target family.

This decomposition specifies what a learning procedure can improve without changing the interface budget. It can produce a code whose information is more accessible to the requested targets, or a decoder that extracts more of the available single-target information. Improving one objective need not improve another. A representation sufficient for a joint function, such as parity, can be poor for random access. Likewise, an increase in average accuracy need not saturate the total information score. Neural-IC combines the causal constraint with this target-dependent organization of retained information. The resulting description goes beyond assigning a bit count to a hidden layer.

The capacity certificate adds a separate slack, $C_H - I(\vec a:Y) \ge 0$. A shortfall of the score from $C_H$ can thus arise from an unattained capacity bound, joint information inaccessible to individual-target queries, or decoder loss. The observed score alone does not identify which mechanism dominates. The trained examples later in Sec.~\ref{sec:experiments} and Appendix~\ref{app:learned-gap} examine
this distinction without assuming that their learned encoders maximize every information objective. Appendix~\ref{app:gap-structure} gives the corresponding posterior interpretation and equality conditions.

\subsection{Nonuniform and correlated data}\label{sec:correlated-data} 

The independent random bits isolate the capacity question from redundancy. For correlated data, the unconditioned score can count the same information repeatedly: if all targets equal one random bit, one communicated bit answers every query and gives an unconditioned sum of $N$. A conditional score avoids this particular double counting~\cite{Jain2024}. For any fixed ordering of a binary database $A^N$, let us define
\begin{eqnarray}
I_{\mathrm{cN\text{-}RAC}} := \sum_K I(A_K:\beta\mid b=K,A_{<K}).
\label{eq:conditional-score}
\end{eqnarray}
The chain rule now requires no independence between target bits:
\begin{eqnarray}
I(A^N:H,B) &=& \sum_K I(A_K:H,B\mid A_{<K}) \nonumber \\
	&\ge& I_{\mathrm{cN\text{-}RAC}}.
\label{eq:conditional-embedding}
\end{eqnarray}
Thus, the same interface certificate bounds the conditional score. Binary Fano bounds further give
\begin{eqnarray}
H(A^N)-\sum_K h(P_{e,K}) \le I_{\mathrm{cN\text{-}RAC}} \le C_H,
\label{eq:conditional-accuracy-bound}
\end{eqnarray}
where $P_{e,K}=\Pr[\beta\ne A_K\mid b=K]$ and the query remains independent of the database. Eq.~(\ref{eq:conditional-score}) is an evaluator-side statistic. It conditions on preceding labels without
giving them to the actual decoder. Its value may depend on their ordering, which must therefore be reported. It need not equal the unconditioned score even when the database bits are independent. For the two-bit parity code in Sec.~\ref{sec:information-gap}, answering with the parity has unconditioned score zero, while conditioning the second target on the first gives a conditional score of one. It is a conditional extension of the capacity law, not an ordering-independent measure of every representation's usefulness.

\section{Quantum correlation layers and nested access}\label{sec:neural-chsh}

A classical representation may be assisted by a quantum resource without communicating a quantum state. This distinction matters: the resource can change the target-to-answer channels at fixed message length, while the information law still bounds their total score. Quantum mutual information
\begin{eqnarray}
I(U:V) = S(\hat{\rho}_U) + S(\hat{\rho}_V) - S(\hat{\rho}_{UV})
\label{eq:quantum-mutual-information}
\end{eqnarray}
obeys the needed chain rule and data processing. Here, $S$ denotes the von Neumann entropy of the indicated density operator. Nonnegativity of its conditional form follows from strong subadditivity~\cite{LiebRuskai1973}. This supplies the information calculus for quantum-assisted representations.

To make the effect on query-conditioned access explicit, consider an isotropic CHSH resource with binary inputs $(s,t)$ and outputs $(u,v)$:
\begin{eqnarray}
P(u,v\mid s,t)=\frac{1}{4}\left[1+(-1)^{u\oplus v\oplus st}E\right],
\label{eq:isotropic-cell}
\end{eqnarray}
where $E\in[0,1]$. For two database bits, the encoder inputs $s=a_0\oplus a_1$ and communicates $x=a_0\oplus u$; the decoder inputs $t=b$ and outputs $\beta=x\oplus v$. For either query the success
probability is $(1+E)/2$. The local resources allow $E\le1/2$, and quantum resources allow $E\le1/\sqrt2$~\cite{Tsirelson1980}. The cell specifies a correlation-assisted encoding and decoding module; it does not establish an advantage for a general learning task.

The nested IC construction uses independent cells to address $N=2^n$ targets with one communicated bit~\cite{Pawlowski2009}. Each query selects a path of $n$ cells. The final error is the parity of their
independent errors, so the conditional biases multiply and
\begin{eqnarray}
&& P_K = \frac{1+E^n}{2}, \nonumber \\
&& I_{\mathrm{N\text{-}RAC}}(n,E) = 2^n\left[1-h\left(\frac{1+E^n}{2}\right)\right].
\label{eq:VI_closed_form}
\end{eqnarray}
The exact information expression uses the binary symmetry of the isotropic protocol. This construction consumes $N-1$ independent cells or shared entangled pairs per episode, although the selected decoder path has depth $n$. A small communicated representation therefore does not imply a small total resource cost or an efficient learned network.

For $0<E<1$, the small-bias expansion gives
\begin{eqnarray}
I_{\mathrm{N\text{-}RAC}}(n,E) \sim \frac{(2E^2)^n}{2\ln2}.
\label{eq:critical-scaling}
\end{eqnarray}
Hence, the score vanishes below $1/\sqrt2$, approaches $1/(2\ln2)$ at the quantum endpoint, and grows without bound above it. In the last case it eventually exceeds the one-bit IC budget. This reproduces the established IC threshold as a stability condition on the family of representation-mediated query channels. For $E<1$, however, $P_K \to 1/2$ with depth, even in the supercritical regime; growing total information is not perfect retrieval of an arbitrarily large database. The maximally nonlocal case $E=1$ instead has $P_K=1$.

The asymmetry is handled by the same construction when each fresh cell has a parity-error bias $E_t$ independent of its encoder input and preceding history, for $t=0,1$. Query-averaged seed accuracies alone do not imply this compositional condition. Under it, a query string $b=(b_1,\ldots,b_n)$ has
\begin{eqnarray}
P_b=\frac{1+\prod_{\ell=1}^n E_{b_\ell}}2.
\label{eq:asym-success-path}
\end{eqnarray}
For binary-symmetric branches, summing $1-h(P_b)$ is the exact score and in general, it is a lower bound. The asymmetric IC analysis gives
\begin{eqnarray}
I_{\mathrm{N\text{-}RAC}} \ge \frac{(E_0^2+E_1^2)^n}{2\ln2},
\label{eq:asym-growth-lower-bound}
\end{eqnarray}
so $E_0^2 + E_1^2 > 1$ suffices for an eventual violation. Branch biases thus retain information about unequal query performance that is lost in a single averaged correlation parameter.

A tunable quantum realization uses the Bell state $\ket{\Phi^+}=(\ket{00}+\ket{11})/\sqrt2$, encoder observables $\hat{A}_0=\hat{\sigma}_z$, $\hat{A}_1=\hat{\sigma}_x$, and decoder observables
$\hat{B}_{0,1}=\cos\varphi\,\hat{\sigma}_z\pm\sin\varphi\,\hat{\sigma}_x$. The local randomization preserving CHSH produces isotropic bias $E_{\mathrm{iso}}=(\cos\varphi+\sin\varphi)/2$, which reaches
$1/\sqrt2$ at $\varphi=\pi/4$. The visibility loss modeled by $E_{\mathrm{eff}}=\nu E_{\mathrm{iso}}$ with $0\le\nu\le1$ moves the protocol below this endpoint. These parameters specify physical
correlation strength; their trainability alone makes no claim about generalization or optimization.

The main comparison keeps the communicated interface classical. If an $m$-qubit system is communicated without prior entanglement with the decoder, its information about classical data is at most $m$ bits. With decoder-side entanglement, that budget can reach $2m$ bits by superdense coding. Quantum communication must therefore be assigned its appropriate resource-assisted capacity rather than identified with an $m$-bit classical interface~\cite{Holevo1973,BennettWiesner1992}.

\section{Capacity and fair query-conditioned accuracy}\label{sec:classical-benchmarks}

The total-score bound does not itself demonstrate quantum advantage. For $0 \le m \le N$, a classical $m$-bit representation that copies $m$ targets and guesses the rest attains $I_{\mathrm{N\text{-}RAC}}=m$. It concentrates its information on particular queries. To compare access across the query set, we instead consider uniformly averaged accuracy and protocols whose accuracies are equal for all queries.

For one classical bit, the known exact optimum is~\cite{Ambainis2009}
\begin{eqnarray}
P^{\mathrm{cl}}_{\mathrm{avg}}(N,1) = \frac12+2^{-N}\binom{N-1}{\lfloor(N-1)/2\rfloor}.
\label{eq:classical-onebit-optimum}
\end{eqnarray}
A short proof fixes the comparison. Write $X_i=(-1)^{a_i}$ and a deterministic message $M=f(X)\in\{\pm1\}$. For query $i$, the best nonconstant decoder has bias $|\mathbb E[X_if(X)]|$; a constant decoder has zero bias. Consequently,
\begin{eqnarray}
\sum_i \abs{\mathbb E[X_if(X)]} &=&\max_{s_i=\pm1}\mathbb E\left[f(X)\sum_i s_iX_i\right] \nonumber \\
	&\le& \mathbb E \abs{\sum_i X_i}.
\label{eq:majority-optimality}
\end{eqnarray}
Majority encoding attains this bound. Counting the binomial distribution of $\sum_iX_i$ yields Eq.~(\ref{eq:classical-onebit-optimum}); random protocols cannot improve the averaged optimum because accuracy is affine in the protocol. An unbiased random choice at a majority tie makes the induced channels binary symmetric and identical across queries. Thus, this optimum is attainable with uniform query accuracy, without requiring uniform success for every individual database string.

Stirling's formula gives
\begin{eqnarray}
P^{\mathrm{cl}}_{\mathrm{avg}}(N,1) = \frac12+\frac1{\sqrt{2\pi N}}+O(N^{-3/2}).
\label{eq:classical-onebit-asymptotic}
\end{eqnarray}
At $E=1/\sqrt2$, the nested quantum-assisted protocol has $P^{\mathrm q}=1/2+1/(2\sqrt N)$ for $N=2^n$. Both advantages above chance scale as $N^{-1/2}$, but the quantum coefficient is larger by $\sqrt{\pi/2}$. For the symmetric majority channels and the nested quantum channels, respectively,
\begin{eqnarray}
I^{\mathrm{cl}}_{\mathrm{majority}} \longrightarrow \frac1{\pi\ln2},
\quad
I^{\mathrm q}_{\mathrm{nested}} \longrightarrow \frac1{2\ln2}.
\label{eq:fair-access-limits}
\end{eqnarray}
These limits compare the named fair-access protocols. They do not make majority encoding optimal for the total information score. Both remain below one, whereas storing a single target saturates one. The quantum enhancement is an improvement in uniformly distributed retrieval accuracy under the same communicated capacity, not an increase of the total budget or a demonstrated quantum learning speedup.

\section{Learned representations under controlled information flow}\label{sec:experiments}

We trained classical encoder--decoder models on the uniform $N=8$ task to test the representation bound in systems whose codes are learned rather than prescribed. The encoder maps eight bipolar database inputs to 32 tanh units and then to $m \in \{1,2,4\}$ binary units. The decoder receives these $m$ units and an eight-dimensional one-hot query, uses 32 tanh units, and produces a scalar logit. The query is withheld from the strict encoder. The forward bottleneck is the sign of its preactivation $z$. Backpropagation replaces the derivative of the sign by that of $\tanh z$. We initialize each weight with an independent zero-mean Gaussian of variance equal to the reciprocal input width and set every bias to zero. All parameters are updated jointly with binary cross entropy and Adam, using learning rate $0.003$, batch size 256, coefficients $(0.9,0.999)$, and denominator offset $10^{-8}$. Each minibatch contains newly sampled independent unbiased databases and uniform queries. We use 3000 updates for each of the fixed seeds $11,23,37,53,71$. No validation-based selection or early stopping is applied.

After training, the parameters are fixed. The evaluation enumerates all $2^8$ databases and all eight queries with equal weight. The answer is obtained by thresholding the decoder logit at zero. We compute the information score from the resulting per-query $2 \times 2$ joint probability tables. This uses the full answer distribution and does not replace the learned model by a binary symmetric answer channel. The evaluation is exact for the specified input distribution, so the reported standard deviations describe differences among training runs rather than finite-sample uncertainty.

To test a physical interface constraint, we also pass each trained binary code through independent bit flips of probability $p$. The decoder is kept fixed. Every one of the $2^m$ noise masks is included with its exact probability. The capacity certificate for the received code is
\begin{eqnarray}
C_H=m[1-h(p)].
\label{eq:learned-bsc-capacity}
\end{eqnarray}
The noise is independent of the database, query, and model parameters. A separate control supplies the query to a one-bit encoder as eight additional one-hot inputs while retaining the same hidden widths and training rule. This deliberately removes the causal separation required by the representation bound.

\begin{figure}[t]
\centering
\includegraphics[width=0.95\columnwidth]{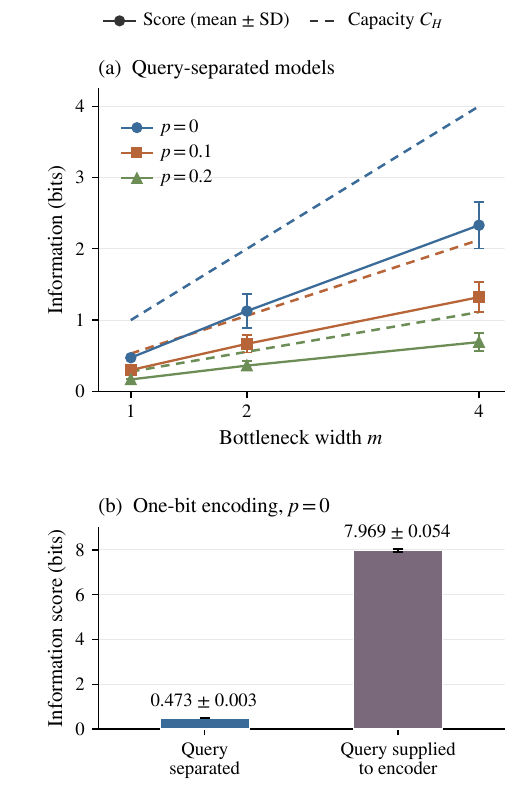}
\caption{Learned retrieval scores and certified capacities from Tab.~\ref{tab:learned-capacity}. (a) Markers show mean scores for query-separated models, and error bars show sample standard deviations over five training seeds. Dashed lines in matching colors give the independently certified capacities. Lines connecting markers guide the eye. (b) Noiseless one-bit models with and without query separation. The query-supplied control forms a different representation for each query, so the query-separated capacity certificate does not apply.}
\label{fig:learned-capacity}
\end{figure}

\begin{table}[t]
\centering
\begin{adjustbox}{width=1.00\columnwidth}
\setlength{\tabcolsep}{6pt}
\begin{tabular}{lcccc}
\toprule
Encoder & $m$ & $p$ & $C_H$ & $I_{\mathrm{N\text{-}RAC}}$ \\
\midrule
Query-separated & 1 & 0 & 1.000 & $0.473\pm0.003$ \\
Query-separated & 2 & 0 & 2.000 & $1.127\pm0.238$ \\
Query-separated & 4 & 0 & 4.000 & $2.331\pm0.326$ \\
Query-separated & 1 & 0.1 & 0.531 & $0.300\pm0.002$ \\
Query-separated & 2 & 0.1 & 1.062 & $0.667\pm0.122$ \\
Query-separated & 4 & 0.1 & 2.124 & $1.322\pm0.211$ \\
Query-separated & 1 & 0.2 & 0.278 & $0.168\pm0.001$ \\
Query-separated & 2 & 0.2 & 0.556 & $0.362\pm0.063$ \\
Query-separated & 4 & 0.2 & 1.112 & $0.692\pm0.126$ \\
Query supplied to encoder & 1 & 0 & $\mathrm{n/a}$ & $7.969\pm0.054$ \\
\bottomrule
\end{tabular}
\end{adjustbox}
\caption{Exact evaluation of learned representations for $N=8$. Scores are means with sample standard deviations over five fixed seeds. $\mathrm{n/a}$ marks a control to which the query-separated capacity certificate does not apply.}
\label{tab:learned-capacity}
\end{table}

Every strict model obeys its certificate, including every individual seed and noisy condition. All five noiseless one-bit models attain the exact classical average-accuracy optimum in Eq.~(\ref{eq:classical-onebit-optimum}), $P^{\mathrm{cl}}_{\mathrm{avg}}(8,1)=0.63671875$. Their information scores are not claimed to be optimal, and their individual-target channels need not be binary symmetric. Without noise, the mean accuracies for $m=1,2,4$ are $0.6367$, $0.6719$, and $0.7531$. Increasing the alphabet therefore enables better retrieval, but the learned models do not saturate its information capacity. The noise lowers the observed score together with the independently specified capacity (see Tab.~\ref{tab:learned-capacity} and Fig.~\ref{fig:learned-capacity}). By contrast, giving the query to the one-bit encoder yields mean accuracy $0.9995$ and an information score close to eight bits. Three seeds recover every database bit without error, whereas the other two retain small errors. The almost eight-bit aggregate score does not describe eight bits stored in one fixed representation. The encoder produces a different representation for each requested bit. This comparison exhibits the role of causal separation in the representation theorem without interpreting optimization quality as a capacity certificate. The retained information and the three components of the budget shortfall are evaluated separately in Appendix~\ref{app:learned-gap}. 

\section{Conclusion and Discussion}\label{sec:conclusion}

We asked how much information a learned representation can make available when its future query is unknown during encoding. The answer requires specifying both the information path and the physical capacity of the interface. We formulated this setting through Neural-RAC and derived the representation law from this causal information structure. It bounds the aggregate information exposed by later answers. Its decomposition distinguishes capacity slack, joint information across targets, and loss during answer extraction. Conditional scores and an explicit memory correction extend the analysis to correlated targets and additional data paths. We then examined how this structure appears in trained binary codes and in a quantum realization of nested access. Exact evaluation of the learned codes shows that almost filling the information budget need not yield comparably informative answers. In the nested CHSH construction, quantum correlations improve uniformly averaged retrieval over the optimal classical one-bit strategy without enlarging the communicated budget. At the Tsirelson limit, finite aggregate information is maintained as each individual answer becomes weak. The physical resource therefore changes the distribution of access across queries, with its resource bank counted separately from the communicated bit.

The central implication for representation learning is that preserving data and making them accessible are separate design problems. A representation can store information jointly across targets while exposing little through individual answers. Even information available about one target can be reduced by the decoder's output alphabet. The gap decomposition gives these mechanisms distinct quantitative meanings. It therefore connects the representation's internal organization to the questions a model must answer and to the form of its outputs. This can help distinguish an encoding constraint from a constraint on answer extraction before redesigning the architecture. The trained examples illustrate this distinction within a controlled task. An information gap does not by itself imply incomplete training, since information score and retrieval accuracy express different objectives.

As a practical design principle, the framework starts with the queries a memory must support and the information paths permitted during use. For compressed memories and encoder--decoder systems satisfying this causal structure, precision and noise specify the budget, while the decomposition identifies where target access is limited. This can inform memory allocation, target selection during training, and the choice of decoder outputs. Additional memories enter the same resource accounting. The resulting theory offers a basis for seeking performance gains through changes matched to the observed limitation, making better use of a fixed budget before increasing the resources assigned to it.

\begin{acknowledgments}
We thank Prof.~Marek {\.Z}ukowski for support. This work was supported by the Korean ARPA-H Project through the Korea Health Industry Development Institute (KHIDI), funded by the Ministry of Health \& Welfare, Republic of Korea (RS-2025-25456722), the Ministry of Trade, Industry and Resources (MOTIR), Korea, under the project ``Industrial Technology Infrastructure Program'' (RS-2024-00466693), and the Ministry of Science, ICT and Future Planning (MSIP) by the National Research Foundation of Korea (RS-2024-00432214, RS-2025-03532992). This work is partially carried out under IRA Programme, project no.~FENG.02.01-IP.05-0006/23, financed by the FENG program 2021-2027, Priority FENG.02, Measure FENG.02.01., with the support of the FNP. M.P. acknowledges support from the NCN Poland, ChistEra2023/05/Y/ST2/00005 under the project Modern Device Independent Cryptography (MoDIC). We acknowledge the Yonsei University Quantum Computing Project Group for providing support and access to the Quantum System One (Eagle Processor), which is operated at Yonsei University.
\end{acknowledgments}

\appendix

\section{Recursive implementation and conditional noise}\label{app:nested-chsh}

This appendix specifies the nested IC protocol in a form that can be implemented without enumerating its successful and unsuccessful database strings~\cite{Pawlowski2009}. It also states the independence condition needed when cell inputs depend on earlier outputs. The calculation concerns the resource model used in the main text; correlated hardware errors need not obey the same product law.

Label the leaves of a depth-$n$ binary tree by words $z\in\{0,1\}^n$, and set $x_z=a_z$. An internal node is a word $w$ of length $r<n$, with children $w0,w1$. Each internal node has one independent isotropic cell with encoder and decoder outputs $(u_w,v_w)$. The encoder evaluates nodes from the leaves upward:
\begin{eqnarray}
s_w=x_{w0}\oplus x_{w1},
\quad
x_w=x_{w0}\oplus u_w.
\label{eq:app-recursive-encoder}
\end{eqnarray}
Only $x_{\varnothing}$, the bit at the root, is transmitted. The encoder does not use the future query. The database and all fixed model choices are independent of the fresh resource bank for this episode.

For the query word $b=b_1\cdots b_n$, the decoder starts with $\widehat x_{\varnothing} = x_{\varnothing}$. At the visited node $w=b_1\cdots b_r$, it inputs $t_w=b_{r+1}$ and updates
\begin{eqnarray}
\widehat x_{wt_w}=\widehat x_w\oplus v_w.
\label{eq:app-recursive-decoder}
\end{eqnarray}
After $n$ steps it outputs $\beta=\widehat x_b$. Unvisited decoder cells can be ignored. There are $2^n-1$ internal nodes, so the small message and the short selected path must be distinguished from the
size of the resource bank and the encoder's complete computation.

For a fresh isotropic cell of bias $E$, its joint distribution has the probability representation
\begin{eqnarray}
u_i=\xi_i,
\quad
v_i=\xi_i\oplus s_it_i\oplus e_i,
\label{eq:app-cell-sampler}
\end{eqnarray}
where $\xi_i$ is an unbiased bit and $\Pr[e_i=0]=(1+E)/2$, independently of $\xi_i$. Drawing independent pairs $(\xi_i,e_i)$ reproduces the independent bank of conditional distributions. This is a joint probability model for analysis.

Let $\mathcal F_{i-1}$ contain the database and the previously revealed variables used to choose the current inputs. An unused independent cell satisfies
\begin{eqnarray}
\Pr[e_i=0\mid\mathcal F_{i-1},s_i,t_i]=\frac{1+E}{2}.
\label{eq:app-fresh-noise}
\end{eqnarray}
This is the conditional independence required by adaptive composition. It is not a claim that an earlier error remains independent of all future inputs: a later input may itself be chosen using an earlier output. Conditioning on such future choices cannot be added to Eq.~(\ref{eq:app-fresh-noise}) without justification.

The local algebra gives
\begin{eqnarray}
x_w\oplus v_w = x_{w0}\oplus s_wt_w\oplus e_w = x_{wt_w}\oplus e_w.
\label{eq:app-child-identity}
\end{eqnarray}
Applying this identity down the selected path $\mathcal P(b)$ yields
\begin{eqnarray}
\beta=a_b\oplus\bigoplus_{w\in\mathcal P(b)}e_w.
\label{eq:app-global-parity}
\end{eqnarray}
For fixed $b$, the path is selected independently of the resource errors. In the explicit independent-cell model, its error signs obey $\mathbb E[(-1)^{e_w}]=E$. Multiplying them, or applying conditional expectation successively in a valid cell-generation order, therefore gives
\begin{eqnarray}
\mathbb E\left[(-1)^{\beta\oplus a_b}\mid b\right]=E^n,
\quad
P_b=\frac{1+E^n}{2}.
\label{eq:app-bias-product}
\end{eqnarray}
The error is independent of the target, establishing the binary-symmetric channel used for the exact information score. If the conditional correctness bias instead equals $E_t$ for a decoder input $t$, independently of the encoder input and previous history, the same calculation gives $\prod_{\ell=1}^nE_{b_\ell}$. A query-averaged seed accuracy by itself does not establish that stronger conditional-noise assumption.

\section{Estimating the observable information}\label{app:computing-score}

The evaluation uses independent fresh episodes after the maps and weights have been fixed. The sampled database must not be reused as an episode-dependent training set. For each query $K$, record the four counts $n^{(K)}_{uv}$ for $(a_K,\beta)=(u,v)$, with $T_K=\sum_{u,v}n^{(K)}_{uv}$. One can fix the query in separate batches or condition a uniformly queried sample on $b=K$. The corresponding
population score uses
\begin{eqnarray}
I_K=\sum_{u,v}p^{(K)}_{uv} \log\frac{p^{(K)}_{uv}}{p^{(K)}_{u\cdot}p^{(K)}_{\cdot v}},
\quad
I_{\mathrm{N\text{-}RAC}}=\sum_K I_K.
\label{eq:app-contingency-score}
\end{eqnarray}
The convention $0\log0=0$ applies. Replacing the probabilities by $n^{(K)}_{uv}/T_K$ gives a descriptive plug-in estimate, not a finite-sample certificate of a capacity violation.

For unbiased target bits, an equivalent population formula is useful for numerical evaluation. Write
$q_K=\Pr[\beta=1\mid a_K=0,b=K]$ and $r_K=\Pr[\beta=1\mid a_K=1,b=K]$. Then,
\begin{eqnarray}
I_K=h\left(\frac{q_K+r_K}{2}\right) - \frac{h(q_K)+h(r_K)}{2},
\label{eq:app-binary-channel-score}
\end{eqnarray}
whereas $P_K=(1-q_K+r_K)/2$. Setting $q_K=1-P_K$ and $r_K=P_K$ recovers $I_K=1-h(P_K)$ for a binary-symmetric channel. Otherwise, the same accuracy can describe different mutual informations. The known analytic channels should therefore be evaluated from their probabilities. The unknown learned channels require either contingency tables or explicitly conservative accuracy bounds.

When enumeration is impractical, random database generation replaces the sum over all $2^N$ strings. It does not remove the statistical cost of resolving weak target--answer correlations. At full support, the leading upward bias of a binary plug-in mutual information is approximately $1/(2T_K\ln2)$~\cite{Paninski2003}. Summing many small estimates can make this bias comparable to the entire information
budget. Sparse tables require particular care; the approximation is not a uniform error bound. Pseudocounts or analytic bias corrections can help descriptive estimates, but they do not by themselves provide a valid confidence lower bound.

For an exact-score confidence region without a symmetry assumption, each contingency-table cell can instead be bounded directly. When $T_K>0$, set
\begin{eqnarray}
\epsilon_K &=& \sqrt{\frac{\ln(8N/\alpha)}{2T_K}}, \nonumber \\
\mathcal Q_K &=& \left\{p\in\Delta_4: \abs{p_{uv}-\widehat p^{(K)}_{uv}} \le \epsilon_K\right\},
\label{eq:app-multinomial-region}
\end{eqnarray}
where $\Delta_4$ is the four-outcome probability simplex. Hoeffding's inequality and a union bound over the $4N$ cells give simultaneous coverage at least $1-\alpha$~\cite{BoucheronLugosiMassart2013}. Minimizing the Shannon mutual information over each region and summing the minima gives a conservative score lower bound. This minimization must bound the global minimum, rather than merely return a local numerical solution. The construction is deliberately simple and may be loose; it illustrates how uncertainty in all four probabilities can be incorporated. A query with no samples contributes only the trivial lower bound zero. It does not support interpolation from other queries unless their symmetry has been established independently.

A conservative test can instead use the accuracy inequality without assuming binary symmetry. If $[L_K,U_K]$ is a simultaneous confidence interval for every $P_K$, then, with the stated joint coverage,
\begin{eqnarray}
I_{\mathrm{N\text{-}RAC}} \ge \sum_K\min_{p\in[L_K,U_K]}[1-h(p)].
\label{eq:app-information-lower-ci}
\end{eqnarray}
Per-query intervals with failure probability $\alpha/N$ give joint coverage at least $1-\alpha$ by the union bound. For independent episodes, exact binomial intervals are available; Wilson intervals are useful approximate alternatives~\cite{BrownCaiDasGupta2001}. If an interval crosses $1/2$, its information lower bound is zero. Otherwise the minimum occurs at the endpoint closest to $1/2$.

Uniform random queries also permit a pooled test when separate query tables would be too sparse. Let $\overline P=N^{-1}\sum_KP_K$. The convexity of $1-h(p)$ and the per-query accuracy bound imply
\begin{eqnarray}
I_{\mathrm{N\text{-}RAC}} \ge N[1-h(\overline P)].
\label{eq:app-pooled-lower-bound}
\end{eqnarray}
Across independent episodes the pooled success count is binomial with parameter $\overline P$. A confidence interval $[L,U]$ thus gives the lower bound $N\min_{p\in[L,U]}[1-h(p)]$, even if query accuracies are unequal and the channels are asymmetric. It is an exact score only when the additional symmetry conditions hold. The correlated episode noise requires an appropriate dependence-aware interval instead of an independent-binomial calculation.

The appropriate violation criterion compares a confidence lower bound for the score with the certified capacity upper bound. Exceeding that upper bound supplies evidence that at least one model assumption or capacity specification fails; it does not identify which one. Near $E=1/\sqrt2$, the nested accuracy advantage is $O(N^{-1/2})$, so resolving it at fixed relative precision requires $O(N)$ independent samples. The closed-form curves should consequently be distinguished from empirical estimates. A reproducible report states the data generator, query distribution, resource independence, fixed weights, interface certificate, sampling design, and score uncertainty.

For the correlated-data score, the evaluator must additionally record the preceding target labels in the declared ordering. The estimation now conditions on $(b,A_{<K})$, and the number of possible conditioning events can grow exponentially with $K$. The chain-rule capacity law does not guarantee that this conditional statistic is easy to estimate from an arbitrary natural dataset. Small controlled generators permit the exact probabilities. Otherwise one needs adequate conditional sample counts or a justified statistical model, with its uncertainty reported. Pooling across different preceding-label values without accounting for that conditioning estimates a different score. The same warning applies if the label ordering is chosen after inspecting the evaluated data: a reproducible comparison fixes the ordering as part of its protocol, rather than selecting whichever value is most favorable.

\section{Gaussian interfaces with classical side information}\label{app:gaussian-certificate}

In this appendix section, we verify the Gaussian certificate for a classical decoder-side register $B$. Let $U \in \mathbb R^d$ be the encoder output and $H=U+Z$, with $Z\sim\mathcal N(0,\sigma^2\mathds{1}_d)$ independent of $(\vec a,U,B)$. Assume an unconditional average-power constraint $\mathbb E\|U\|^2\le dP$ and independence of the fresh database from $B$. The decoder receives $H$ and $B$, but no separate realization of $Z$. All differential entropies below are measured in bits and denoted by $h_{\mathrm d}$, to distinguish them from binary entropy.

Because $H$ is conditionally independent of the database given $(U,B)$, the conditional data processing gives
\begin{eqnarray}
I(\vec a:H,B) &=& I(\vec a:H\mid B) \nonumber \\
	&\le& I(U:H\mid B) \nonumber \\
	&=& h_{\mathrm d}(H\mid B)-h_{\mathrm d}(Z) \nonumber \\
	&\le& h_{\mathrm d}(H)-h_{\mathrm d}(Z).
\label{eq:app-gaussian-side-information}
\end{eqnarray}
No separate power bound is imposed for each value of $B$. The unconditional constraint suffices because $\mathrm{tr}\,\mathrm{Cov}(H)\le d(P+\sigma^2)$. The Gaussian maximum entropy and the arithmetic--geometric mean bound on covariance eigenvalues give
\begin{eqnarray}
h_{\mathrm d}(H) &\le& \frac d2\log[2\pi e(P+\sigma^2)], \nonumber \\
h_{\mathrm d}(Z) &=& \frac d2\log(2\pi e\sigma^2).
\label{eq:app-gaussian-entropies}
\end{eqnarray}
Their difference proves Eq.~(\ref{eq:awgn-capacity-certificate}) under the stated assumptions~\cite{CoverThomas2006}. The noise known to the decoder, a side channel carrying the database, or a missing power constraint describes a different interface and invalidates this certificate.

For quantum $B$, ordinary conditional differential entropy is not the correct formalism: one must use a classical--quantum state and justify the corresponding conditional-information bound. The representation theorem accommodates such a resource whenever an independent capacity certificate has been established. The elementary calculation above specifically certifies the classical-side-information Gaussian model; it should not be applied to an unspecified quantum noise or resource model.

\section{Information structure and saturation of the representation bound}\label{app:gap-structure}

The gap decomposition in Sec.~\ref{sec:information-gap} has a statistical interpretation when the available system $Y=(H,B)$ is classical. This interpretation distinguishes information organized across several targets from information discarded by the final answers. Write $\beta_K$ for the decoder output on query $K$. The equalities below are consequences of the chain rule and data processing; they do not introduce an additional causality principle.

\subsection{Posterior dependence between independent targets} 

Let $A=\vec a$ have independent components, and evaluate all quantities for the fixed model and resource specification. Define the joint-target term
\begin{eqnarray}
G_{\mathrm{joint}} &:=& I(A:Y)-\sum_K I(a_K:Y) \nonumber \\
	&=& \sum_K H(a_K\mid Y)-H(A\mid Y).
\label{eq:app-posterior-total-correlation}
\end{eqnarray}
The second equality uses the product distribution of the targets before encoding. It expresses the conditional total correlation of their posterior distribution, the conditional form of total correlation~\cite{Watanabe1960}. More explicitly,
\begin{eqnarray}
G_{\mathrm{joint}} = \sum_y p(y)\,D_{\mathrm{KL}}\left( p(A\mid y)\,\middle\|\,\prod_Kp(a_K\mid y) \right).
\label{eq:app-posterior-divergence}
\end{eqnarray}
The summations are replaced by integrals for a continuous classical interface. The nonnegativity of relative entropy proves $G_{\mathrm{joint}} \ge 0$. The equality holds precisely when the targets are conditionally independent given $Y$, for almost every available representation value.

Thus, an initially independent database can have dependent targets after the representation is specified. Their joint posterior then contains structure that cannot be described by its individual-target marginals. Eq.~(\ref{eq:app-posterior-divergence}) identifies that effect for the chosen target family. It should not be interpreted as a unique decomposition of information into redundancy and synergy under all definitions of those concepts. Changing the requested targets can change what counts as joint information, even for the same interface.

The total $G_{\mathrm{joint}}$ is independent of the target ordering, since Eq.~(\ref{eq:app-posterior-total-correlation}) has no order in its right-hand side. Its chain-rule representation as $\sum_K I(a_K:a_{<K}\mid Y)$ redistributes contributions among targets when the order changes. This differs from the ordered conditional retrieval score, which uses different actual decoder branches after conditioning on preceding labels and can itself change with their order.

For independent targets, the conditional retrieval score satisfies
\begin{eqnarray}
I_{\mathrm{cN\text{-}RAC}}-I_{\mathrm{N\text{-}RAC}} = \sum_K I(a_K:a_{<K}\mid\beta_K) \ge 0.
\label{eq:app-conditional-score-increment}
\end{eqnarray}
This identity follows by expanding the branchwise conditional mutual information in two orders. Its increment concerns information exposed after preceding labels are supplied to the evaluator, rather than
improved operational access for a decoder that never receives them. For correlated targets the prior-dependence terms no longer vanish, so this nonnegative comparison need not hold.

\subsection{Decoder sufficiency and the final answer alphabet} 

For each query, denote its output channel by $Y\to\beta_K$. The classical Markov condition $a_K\to Y\to\beta_K$ gives
\begin{eqnarray}
G_{\mathrm{dec},K} &:=&I(a_K:Y)-I(a_K:\beta_K) \nonumber \\
	&=& I(a_K:Y\mid\beta_K) \ge 0.
\label{eq:app-decoder-loss-classical}
\end{eqnarray}
Here, $I(a_K:\beta_K)$ is the branch quantity $I(a_K:\beta\mid b=K)$ in the main text. The equality follows by expanding $I(a_K:Y,\beta_K)$ in two orders and using $I(a_K:\beta_K\mid Y)=0$.

The loss vanishes if and only if the answer is a sufficient statistic for that target: conditioning on the answer leaves no further information about $a_K$ in $Y$. Equivalently,
\begin{eqnarray}
p(a_K\mid y) = p(a_K\mid\beta_K)
\label{eq:app-decoder-sufficiency}
\end{eqnarray}
almost surely with respect to the joint law of $(Y,\beta_K)$. This condition applies to stochastic as well as deterministic decoders.

Decoder loss is not synonymous with incomplete optimization. The required output alphabet can prevent sufficiency. A binary answer cannot generally preserve a posterior probability with more than two
essentially distinct values. Even an accuracy-optimal decoder that thresholds its posterior can therefore lose information by grouping different posterior values into the same answer. A trained model's
loss must be interpreted relative to its permitted answers and objective; it need not disappear after better training. The information law constrains the available representation, whereas its measured retrieval score additionally depends on this output choice.

\subsection{Exact equality conditions and resource fixation} 

Combining the terms gives the full budget shortfall
\begin{eqnarray}
C_H - I_{\mathrm{N\text{-}RAC}} = G_{\mathrm{cap}} + G_{\mathrm{joint}} + G_{\mathrm{dec}}.
\label{eq:app-capacity-saturation}
\end{eqnarray}
Here, $G_{\mathrm{cap}} = C_H - I(A:Y)$ and $G_{\mathrm{dec}}=\sum_KG_{\mathrm{dec},K}$. All terms are nonnegative. The retrieval score saturates its certificate if and only if the interface attains the stated bound, the posterior targets factorize, and every decoder branch is sufficient for its target. A conservative capacity upper bound need not be attainable. Its slack is consequently not evidence of an inadequate learning procedure.

For a classical $m$-bit interface with classical $B$, independence of the database from $B$ also gives
\begin{eqnarray}
m-I(A:H,B) &=& m - H(H) + I(H:B) \nonumber \\
	&& + H(H\mid A,B).
\label{eq:app-hard-capacity-slack}
\end{eqnarray}
Attaining the $m$-bit budget requires a uniform full-alphabet message, independence of that message from $B$, and no message uncertainty remaining after $(A,B)$ is specified. These are capacity-attainment conditions. The additional joint-target and decoder conditions remain necessary to saturate the retrieval score.

Fixing or including shared parameters is essential throughout. For a fresh unbiased bit $a$ and an independent shared random bit $W$, the code $H=a\oplus W$ has $I(a:H)=0$, but a decoder with $W$
recovers $a$ exactly and $I(a:H,W)=1$. Independence of the parameters from the data does not justify discarding them from decoder information. Conditioning on their fixed realization, or including them in $B$, restores the stated model and its bound.

For quantum $Y$, the main decomposition remains valid through quantum conditional mutual information and data processing. The equality requires zero capacity slack, zero conditional-information contributions in the joint-target term, and equality in each decoder's information-processing bound. The classical posterior factorization and the Markov identity in Eq.~(\ref{eq:app-decoder-loss-classical}) are not asserted for an unmeasured quantum register. In particular, this discussion does not presume that the original quantum system remains available after producing its classical answer.

\section{Information organization in learned binary codes}\label{app:learned-gap}

The trained codes allow the representation--retrieval gap to be resolved into its separate contributions. For these classical models there is no decoder-side resource $B$. Let $H_0$ be the deterministic binary encoder output and let $H$ denote the received code after independent bit flips of probability $p$. Enumerating the 256 databases and all $2^m$ received codes gives the complete distribution of $(\vec a,H)$. We calculate
\begin{eqnarray}
I(\vec a:H) &=& H(H)-m h(p), \nonumber \\
S_H&=&\sum_{K=0}^{7} I(a_K:H).
\label{eq:learned-retained-information}
\end{eqnarray}
Here, $H(H)$ is the Shannon entropy of the entire received code. The first identity follows because the noiseless code is deterministic and the independent noise contributes conditional entropy $m h(p)$. It is also checked directly from the joint distribution.

The three nonnegative contributions to the capacity--score gap are
\begin{eqnarray}
G_{\mathrm{cap}} &=& C_H-I(\vec a:H), \nonumber \\
G_{\mathrm{joint}} &=& I(\vec a:H)-S_H, \nonumber \\
G_{\mathrm{dec}} &=& S_H-I_{\mathrm{N\text{-}RAC}}.
\label{eq:learned-gap-components}
\end{eqnarray}
For this interface, $G_{\mathrm{cap}}=m-H(H)$. We verify that $G_{\mathrm{joint}}$ also equals $\sum_K I(a_K:\vec a_{<K}\mid H)$ by separately evaluating each conditional-information term. The answer score is obtained from the frozen decoder on every received code and query. Every seed and noise condition satisfies
\begin{eqnarray}
C_H=I_{\mathrm{N\text{-}RAC}}+G_{\mathrm{cap}}+G_{\mathrm{joint}}+G_{\mathrm{dec}}.
\label{eq:learned-full-gap}
\end{eqnarray}
The agreement with the direct answer-channel evaluation is within $10^{-10}$ bits. The query-dependent control is excluded from this decomposition because it has no common representation committed before the query.

\begin{figure}[t]
\centering
\includegraphics[width=0.85\columnwidth]{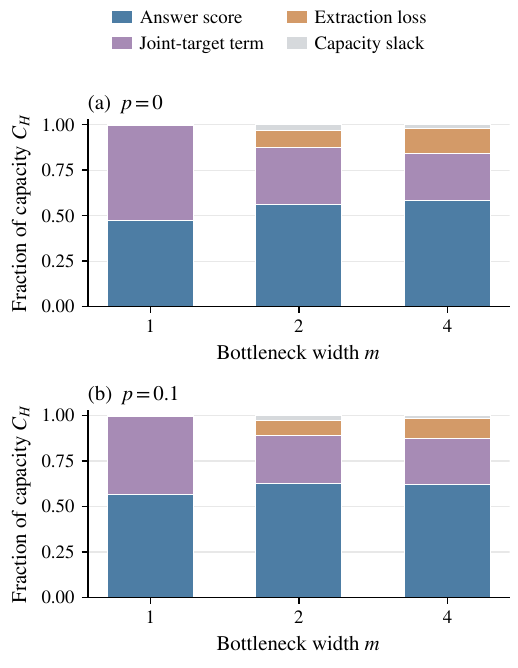}
\caption{Information organization in trained binary representations. Each bar averages five fixed training seeds and divides its four components by the independently certified capacity $C_H=m[1-h(p)]$. Panel (a) evaluates the noiseless codes and panel (b) uses independent bit flips with $p=0.1$ without retraining. The gray segment is the unused part of the certified budget. The purple segment is joint-target information inaccessible to separate target queries. The orange segment is the loss in forming the decoder's binary answers. Tab.~\ref{tab:learned-information-components} reports the absolute contributions and their variation among seeds. Noise lowers the absolute score, even where its normalized fraction increases.}
\label{fig:learned-information-decomposition}
\end{figure}

\begin{table*}[t]
\centering
\begin{adjustbox}{width=0.75\textwidth}
\setlength{\tabcolsep}{10pt}
\begin{tabular}{ccccccc}
\toprule
$m$ & $p$ & $C_H$ & $I_{\mathrm{N\text{-}RAC}}$ & $G_{\mathrm{cap}}$ & $G_{\mathrm{joint}}$ & $G_{\mathrm{dec}}$ \\
\midrule
1 & 0 & 1.000 & $0.473\pm0.003$ & $0.006\pm0.005$ & $0.521\pm0.006$ & $0$ \\
2 & 0 & 2.000 & $1.127\pm0.238$ & $0.067\pm0.069$ & $0.617\pm0.136$ & $0.189\pm0.203$ \\
4 & 0 & 4.000 & $2.331\pm0.326$ & $0.093\pm0.050$ & $1.027\pm0.166$ & $0.549\pm0.150$ \\
1 & 0.1 & 0.531 & $0.300\pm0.002$ & $0.004\pm0.003$ & $0.227\pm0.003$ & $0$ \\
2 & 0.1 & 1.062 & $0.667\pm0.122$ & $0.028\pm0.028$ & $0.280\pm0.054$ & $0.087\pm0.101$ \\
4 & 0.1 & 2.124 & $1.322\pm0.211$ & $0.040\pm0.023$ & $0.538\pm0.106$ & $0.225\pm0.093$ \\
\bottomrule
\end{tabular}
\end{adjustbox}
\caption{Exact information contributions in bits. Values are means with sample standard deviations over the five fixed seeds. Evaluation enumerates all databases, queries, and independent noise outcomes. Standard deviations therefore describe training variability. Small differences in the sums of printed values result from rounding.}
\label{tab:learned-information-components}
\end{table*}

Fig.~\ref{fig:learned-information-decomposition} and Tab.~\ref{tab:learned-information-components} summarize these contributions. The noiseless one-bit codes retain $0.994$ bits on average, close to the one-bit certificate, but expose only $0.473$ bits in their answers. Almost the whole remaining gap is therefore joint-target organization, rather than failure to use the alphabet. Their binary decoder maps are information preserving for each query, so $G_{\mathrm{dec}}=0$. The two- and four-bit codes also nearly fill their retained-information budgets. Their gaps include positive answer-extraction losses, but the joint-target term remains the largest average contribution. Independent noise reduces both the certified capacity and the retained information. At $p=0.1$, the retained amounts are $0.527$, $1.034$, and $2.084$ bits for $m=1,2,4$, respectively. These values remain close to the respective certificates even though the answer scores are substantially smaller.

The decomposition thus identifies a feature that the retrieval scores alone cannot resolve. A learned representation can use almost its entire alphabet while organizing much of the information jointly across targets. The extraction term additionally depends on the imposed binary-answer format and the chosen decoder. Its presence does not by itself diagnose an optimization failure.

\section{Cryptographic sponges as a queryless bottleneck analogue}\label{app:sponge}

Cryptographic sponge constructions give a familiar example of producing a long output through a finite intermediate state~\cite{Bertoni2008Sponge}. They underlie Keccak-based hash and extendable-output functions, including SHA-3 and SHAKE~\cite{NIST2015FIPS202}. A sponge maintains a state of width $w=r+c$ bits. During absorption, padded input blocks enter the $r$-bit rate portion of the state, with applications of a fixed permutation between blocks. Squeezing then returns blocks from this portion while continuing to apply the permutation. The remaining $c$ bits are called the cryptographic capacity. This construction parameter enters the security analysis and differs from the information-capacity certificate $C_H$ used here. For the information accounting below, the retained system is the entire $w$-bit state.

Fix the public permutation, padding rule, and other construction parameters. Let $X$ denote the absorbed input and $S$ the terminal state. For a fixed requested output length $L$, squeezing without further input produces a deterministic string $O_L=g_L(S)$. The Markov relation $X \to S \to O_L$ and the deterministic output map give
\begin{eqnarray}
I(X:O_L) &\le& I(X:S) \le H(S) \le w, \nonumber \\
H(O_L) &\le& \min\{L,H(S)\}.
\label{eq:sponge-data-processing}
\end{eqnarray}
Thus, increasing the output length may reveal more of the retained state, but cannot reveal information about the absorbed input beyond what that state contains. The bound follows from data processing and the state alphabet, independently of the difficulty of reconstructing that state.

The analogy clarifies why downstream output length differs from information crossing an intermediate interface. For $L>w$, the output cannot consist of $L$ independent unbiased bits, since its joint entropy is at most $w$. The output positions should therefore not be counted as new independent input targets in Neural-RAC. The squeezing itself supplies no later query selecting a particular entry of the original input. Neural-RAC additionally specifies that target selection and measures the information exposed by its answers. The cryptographic security concerns distinguishability and the computational or query cost of attacks. Eq.~(\ref{eq:sponge-data-processing}) describes retained information, without supplying a cryptographic security bound.

\bibliography{ref_Neural-IC}

\end{document}